\documentclass[
  reprint,
  superscriptaddress,
  floatfix,
  aps,
  prx,
  amsmath,
  amssymb,
]{revtex4-2}
\usepackage{subfigure} % for subfigures within figures
\usepackage{graphicx}
\usepackage{bm}
\usepackage{braket} % for quantum bra-ket notation
\usepackage[colorlinks=true,allcolors=blue]{hyperref}
\usepackage{amsthm}
\newtheorem{theorem}{Theorem}[section]

\newtheorem{claim}[theorem]{Claim}

\newtheorem{proposition}[theorem]{Proposition}

\theoremstyle{definition} % Flips the switch
\newtheorem{definition}[theorem]{Definition}

\theoremstyle{remark}

\graphicspath{{figs/}}

\begin{document}

\title{\textit{Lifting}: Quantum State Discrimination Framework for Mid-Circuit Measurements}

\author{Chi Zhang}
\email[Contact author: ]{czh4@iu.edu}
\affiliation{Indiana University Bloomington, Bloomington, Indiana, USA}

\author{Neel Vora}
\email[Contact author: ]{NRVora@lbl.gov}
\affiliation{Lawrence Berkeley National Laboratory, Berkeley, California, USA}

\author{Yilun Xu}
\email[Contact author: ]{yilunxu@lbl.gov}
\affiliation{Lawrence Berkeley National Laboratory, Berkeley, California, USA}

\author{Gang Huang}
\email[Contact author: ]{ghuang@lbl.gov}
\affiliation{Lawrence Berkeley National Laboratory, Berkeley, California, USA}

\author{Fan Chen}
\email[Contact author: ]{fc7@iu.edu}
\affiliation{Indiana University Bloomington, Bloomington, Indiana, USA}

\author{Lei Jiang}
\email[Contact author: ]{jiang60@iu.edu}
\affiliation{Indiana University Bloomington, Bloomington, Indiana, USA}

\date{\today}

\begin{abstract}
%!TEX root = main.tex
% abstract 
% smooth, professional prose and confident tone
% never semicolon
% never acronym without explanation. 
% 1. QC is important.
% 2. MCM is important but is indered by noise on NISQ hardware 
% 3. MCM relies on accurate QSD
% 5. previous NN QSD is classification problem and bounded by Bayes risks
% 6. state deterioration is complementary challenge to bayes risk
% 7. how we address 5 and 6 in Lifting and point out rigorous theoretical fundation
% 8. evaluated across 3 IBM processors, achieved up to XX% improvement in representative MCM workloads. (pick highest improvement)

Quantum computing promises new capabilities beyond classical computation. Mid-circuit measurements (MCMs) enable error correction and outcome-dependent control, but their reliability is limited by noise on superconducting hardware. These operations depend on accurate quantum state discrimination (QSD), which assigns measured signals to state labels. Neural network-based (NN-based) QSD improves this assignment, yet its error remains bounded below by the Bayes risk of the available readout distributions. A complementary challenge is state deterioration as successive measurements accumulate and circuit execution advances, making static calibration less representative of the workload. We present \textit{Lifting}, a QSD framework that addresses both challenges through repetition encoding, neural decoding, and \textit{in situ} calibration. Encoding supplies joint in-phase and quadrature (I--Q) readout signals from the target and auxiliary qubits, while contextual calibration trains the decoder for the conditions encountered during execution. Our theoretical analysis shows that retaining additional channels from a fixed joint distribution cannot increase Bayes risk and establishes sufficient conditions for contraction of an overlap-based upper bound. Experiments across three IBM processors evaluate the framework on repeated stabilizer measurement, parity-check, and teleportation workloads. 
\end{abstract}

\maketitle

%!TEX root = main.tex
\section{Introduction}\label{sec:introduction}
%unify terminology
% Neural network-based (NN-based) QSD
% joint I--Q readout signals/distributions/vectors for the encoded signals, rather than nested readout vectors
% Neural decoder/decoding rather than joint-observation
% fix state deterioration for all 
% do not write acronyms without introduction and explanation (e.g. SPAM)
% use \textit{Lifting} for the name of framework, and do not use Lift (in situ) NN (in-situ) in results.tex, use \textit{Lifting} with \textit{in situ} calibration instead

%tone
%be professional, but confident about the claims of contribution and the weakness of the previous works
%transition should be smooth in meaning between every graph
%write clear prose, never write sentences with connected by semicolon (;)

% find the references (at premier vuenues like Nature, Science, PRL or PRX Quantum) that support my claims
% paragraph 1 QC and challenges
% 1. Quantum Computing shows promises that classical computing doesn't. 
% 2. MCM circuits constitue a significant part of SOTA QC. 
% 3. However, NISQ hardware is noisy, and (mid-circuit) measurement is the most error-prone part on superconducting hardware. this hinders MCM thus the premises of QC. 
Quantum computing promises new approaches to computational tasks that are difficult for classical computers, motivating sustained progress in superconducting quantum processors~\cite{aruteQuantumSupremacyUsing2019}. Realizing this promise requires reliable operations throughout a computation. Mid-circuit measurements (MCMs) and real-time feedback are central to this effort, enabling quantum error correction and dynamic circuits that adapt their operations to measurement outcomes~\cite{sivakRealtimeQuantumError2023,corcolesExploitingDynamicQuantum2021}. However, measurement remains a major source of error on superconducting hardware. Its effects extend beyond an incorrect reported outcome to disturbances of the measured qubit and neighboring qubits, limiting the reliability of subsequent computation~\cite{hothemMeasuringErrorRates2025}.

%paragrph 2 MCM and QSD
%1. introduce what MCM is. 
%2. explain how measurement works in MCM on superconducting hardware
%3. relate QSD to classification problems. 
%4. tell people what in accurate QSD can cause 
An MCM measures a qubit before circuit execution ends, allowing its outcome to guide subsequent operations. On superconducting hardware, dispersive readout obtains this outcome by probing a resonator whose response depends on the qubit state. The resulting microwave signal is demodulated into in-phase and quadrature (I--Q) components. Quantum state discrimination (QSD) then assigns these continuous signals to discrete state labels, typically $\ket{0}$ and $\ket{1}$~\cite{blaisCircuitQuantumElectrodynamics2021}. This makes QSD a classification problem at the interface between quantum evolution and classical control. An incorrect assignment can select the wrong feedback operation or produce an incorrect correction frame, allowing a readout error to affect the remainder of the circuit.
% paragraph 3 Neural network-based QSD and bayes risk
% 1. SOTA neural netork-based QSD (cite what you already cited) achievements and strengths. 
%2. however, NN-based QSD is fundamentally bounded by  bayes risk
%3. introduce what Bayes risk is 
%4. and point out Bayes risk is a hard lower bound in discrimination error, regardless of training and architecture of the NN in NN-based QSD
%5. moving beyond bayes risk of single NN needs design rather than NN architecture

Neural network-based (NN-based) QSD has improved this interface in several complementary ways. Lienhard et al. improved multiplexed discrimination by accounting for readout crosstalk~\cite{lienhardDeepNeuralNetworkDiscriminationMultiplexed2022}, Maurya et al. developed a more compact architecture using matched filters and a neural network~\cite{mauryaScalingQubitReadout2023}, and Vora et al. implemented low-latency discrimination on a field-programmable gate array (FPGA) for real-time control~\cite{voraMLPoweredFPGAbasedRealTime2024}. Nevertheless, every discriminator operating on the same readout signals is subject to the same fundamental limit. The Bayes risk is the minimum achievable classification error for given class distributions and prior probabilities. When those distributions overlap, even an optimal classifier must sometimes assign the wrong label. Increasing network capacity or improving training cannot reduce error below this limit~\cite{devroyeProbabilisticTheoryPattern1996}. Reducing error below the Bayes risk of single-qubit readout therefore requires changing the information available to the decoder, rather than improving the neural network alone.

% paragraph 4 state deterioration 
% 1. a complementary challenge is state deterioration. refer to the figure. 
% 2. identify that this is different from the long-term drift as \cite{korhonenPracticalTechniquesHighprecision2025,proctorDetectingTrackingDrift2020}
% 3. only recent literature \cite{zhaoSingleStepPhaseEngineeredPulse2026,santosDriftresilientMidcircuitMeasurement2025} have identified this type of state deterioration. 
% 4. blame this to context (major underlying physics is  residual photons accumulate in resonator cavity over time)
% 5. the context-depent nature of MCM circuits calls out in-situ calibration

A complementary challenge is state deterioration during circuit execution, reflected here in changing I--Q distributions and declining discrimination accuracy. Calibration obtained from isolated state-preparation-and-measurement (SPAM) circuits can consequently become poorly matched to the signals encountered within the workload. Figure~\ref{fig:main_readout} illustrates this behavior in the five-cycle qubit-reuse stabilizer (QR-Stabilizer) workload. In this workload, state deterioration becomes more severe as the number of successive measurements increases and circuit execution advances in time. Later I--Q distributions depart further from the SPAM references, while discrimination accuracy declines. This dependence on circuit history differs from the longer-term parameter drift addressed by calibration-monitoring methods~\cite{proctorDetectingTrackingDrift2020,korhonenPracticalTechniquesHighprecision2025}. Residual photons provide one physical mechanism for such history dependence. Photons remaining in a readout resonator after a measurement can continue to disturb the qubit and interfere with subsequent operations, motivating both established and recent cavity-depletion techniques~\cite{mcclureRapidDrivenReset2016,zhaoSingleStepPhaseEngineeredPulse2026}. Santos and Uzdin also explicitly account for quantum-state degradation during repeated measurements in their mitigation framework~\cite{santosDriftresilientMidcircuitMeasurement2025}. These results motivate treating the measurement context as part of the discrimination problem and acquiring \textit{in situ} calibration that reflects the preceding operations and measurement timing.

% paragraph 5 overview of current paper 
We present \textit{Lifting}, a QSD framework that combines repetition encoding, neural decoding, \textit{in situ} calibration, and optional confidence filtering. Repetition encoding couples the measured qubit to auxiliary qubits and supplies joint I--Q readout signals to the neural decoder. Signals that overlap in a single-qubit I--Q plane can become distinguishable through information carried by the auxiliary channels. This changes the classification problem and creates an opportunity to reduce its Bayes risk. \textit{In situ} calibration addresses the complementary problem of state deterioration by training the decoder on references acquired in the measurement context. Confidence filtering further allows applications to trade retained-shot coverage for more reliable accepted outcomes.

%paragraph 6 identify the contributions
%contributions needs to be general, applicable to the broader field
%do not be obsessed with what we did specifically, because the contributions is to the community
%so do not mention too specific results like distance sweep, QR-stabilizer etc.  
Our contributions connect these principles to practical readout on superconducting hardware. First, we characterize within-circuit state deterioration and establish calibration context as an important consideration for NN-based QSD. Second, we provide a theoretical foundation for decoding joint I--Q readout signals, showing that retaining additional channels from a fixed joint distribution cannot increase Bayes risk and establishing sufficient conditions for contraction of an overlap-based upper bound. Third, we identify a practical encoding choice that balances discrimination accuracy against the cost of auxiliary qubits and encoding operations. Finally, experiments on IBM processors demonstrate how neural decoding and \textit{in situ} calibration improve workload-specific circuit success in representative MCM workloads.

The remainder of this paper is organized as follows. Section~\ref{sec:background} reviews dispersive readout and the Bayes-risk formulation. Section~\ref{sec:phenomenology} examines within-circuit readout-distribution changes. In Section~\ref{sec:framework}, we present the theoretical framework of our joint soft-information protocol and its high-dimensional embedding. Section~\ref{sec:results} details the experimental setup and evaluates the protocol's performance on IBM quantum processors. Finally, Section~\ref{sec:discussion} concludes the paper.

%!TEX root = main.tex
\section{MCM Circuits and the Bayes Risk Bound}\label{sec:background}
This section reviews MCM workloads, dispersive readout, NN-based QSD, and the Bayes-risk limit of a fixed observation space.

\subsection{MCM Circuit}\label{subsec:mcm_circuit}

A mid-circuit measurement is performed before circuit execution ends. Its outcome can guide subsequent gates, update a correction frame, or support repeated syndrome extraction. Reset and reuse provide another application: a measured physical qubit can be reinitialized and assigned a later role, reducing the physical-qubit footprint~\cite{decrossQubitReuseCompilationMidCircuit2023,brandhoferOptimalQubitReuse2023,kimQRMapMapBasedApproach2025}. Reset is therefore part of some MCM workloads, rather than a requirement for every MCM.

Reliable MCM execution depends on both the quantum operations and the interpretation of their readout. An assignment error changes the inferred classical outcome and can lead to incorrect feedback or frame tracking. Reset infidelity concerns the state left by reinitialization. It is a separate error mechanism, even though active reset may use a measured outcome to select a correction. On superconducting devices, measurement backaction, crosstalk, and decoherence during measurement and feedback delays can further affect circuit performance~\cite{zhongCompilationFrameworkQuantum2025,hothemMeasuringErrorRates2025}. The workloads considered here cover repeated measurement and reuse, parity extraction, and teleportation with corrections interpreted through a Pauli frame.

\subsection{Dispersive Readout in Dynamic Circuits}

In superconducting qubit systems, dispersive readout is the standard technique for measuring the state of a qubit. This method relies on the interaction between the qubit and a readout resonator, where the qubit state shifts the resonator frequency, allowing the state to be inferred from the transmitted or reflected microwave signal. The Hamiltonian describing this interaction in the dispersive regime is given by
\begin{equation}\label{eq:dispersive_hamiltonian}
H/\hbar = \omega_r a^\dagger a + \frac{\omega_q}{2} \sigma_z + \chi a^\dagger a \sigma_z,
\end{equation}
where $\omega_r$ and $\omega_q$ are the resonator and qubit frequencies, respectively, $a^\dagger$ and $a$ are the creation and annihilation operators of the resonator mode, $\sigma_z$ is the Pauli-$z$ operator of the qubit, and $\chi$ is the dispersive shift.

After the interaction, the transmitted or reflected microwave signal carries information about the qubit state. The signal is typically amplified, integrated, and demodulated, yielding a voltage signal whose in-phase and quadrature (I--Q) components encode the qubit state information. On IBM hardware, a readout signal is typically a point on the I--Q plane, as shown in Fig.~\ref{fig:main_readout}(c). A discriminator, such as a threshold or machine-learning-based classifier, processes the I--Q signal to infer the qubit state, typically distinguishing between the ground and excited states.

In an MCM circuit, however, Eq.~\eqref{eq:dispersive_hamiltonian} does not fully capture the measurement context. During circuit execution, the measured qubit can be affected not only by intrinsic relaxation during lifetime $T_1$ and dephasing during time $T_2$, but also by mid-circuit measurement backaction, crosstalk from neighboring qubits, and measurement-induced errors. Residual photons in a readout resonator can also disturb subsequent qubit operations, motivating cavity-depletion techniques~\cite{mcclureRapidDrivenReset2016,zhaoSingleStepPhaseEngineeredPulse2026}. Prior drift studies examine changes in device parameters over calibration windows \cite{korhonenPracticalTechniquesHighprecision2025,proctorDetectingTrackingDrift2020}. CaliQEC instead enables \textit{in situ} device calibration during surface-code operation \cite{fangCaliQECInsituQubit2025}. It addresses a different calibration task from learning a readout discriminator at a particular MCM cycle. In MCM circuits, by contrast, state deterioration reflected in these readout changes can occur on microsecond-level timescales within a single circuit execution.

\subsection{Quantum State Discrimination from SPAM Calibration}
For superconducting platforms, QSD converts the measured analog signal into a classical state assignment. Lienhard et al. demonstrate NN-based QSD for multiplexed readout, improving assignment accuracy by accounting for system-dependent nonidealities such as crosstalk~\cite{lienhardDeepNeuralNetworkDiscriminationMultiplexed2022}. Maurya et al. combine hierarchical matched filters with a compact neural network to reduce the hardware cost of discrimination~\cite{mauryaScalingQubitReadout2023}. Vora et al. implement neural inference on an FPGA to support real-time MCM control~\cite{voraMLPoweredFPGAbasedRealTime2024}. These works address complementary accuracy, scalability, and latency challenges.

Our comparison concerns the calibration data supplied to a discriminator. The SPAM baseline in this paper is trained on isolated prepared-state references. The contextual alternative includes preceding workload operations and matches the measurement cycle. The QR-Stabilizer experiment shows that these acquisition conditions can produce different readout distributions. We evaluate this calibration dependence directly in the QR-Stabilizer experiments.

Encoding-based methods provide a complementary approach. Repetition-style readout schemes add auxiliary measurements and decode their hard assignments~\cite{hicksActiveReadouterrorMitigation2022,guntherImprovingReadoutQuantum2021}. Coherent Pauli checks instead use check circuits, including one-sided checks for measurement-error mitigation~\cite{vandenbergSingleshotErrorMitigation2023}. They are not an ordinary repetition-majority decoder. Soft-information decoding retains analog confidence for code-specific decoding~\cite{hanischSoftInformationDecoding2026}. These distinctions motivate comparing hard repetition decoding with neural decoding of joint I--Q readout signals in the expanded observation space.

\subsection{Bayes Risk Bound of Discrimination}\label{subsec:bayes_risk_bound}

Because QSD is a classification problem, a quantitative way to understand this limitation is through the concept of Bayes risk~\cite{devroyeProbabilisticTheoryPattern1996}, which gives the minimum possible error probability for a classification task. In superconducting-qubit readout, this means that any discriminator assigning an analog measurement observation to a quantum-state label has an error rate no lower than the Bayes risk associated with its observation space, regardless of architecture or training strategy.

Formally, we have
\begin{definition}[Bayes risk]\label{app:bayes_risk}
Given random variables $(X,Y)$ taking values in an observation space $\mathcal{X}$ and a label set $\mathcal{Y}$, respectively, let $g^*$ be the \textit{Bayes classifier} defined by~\cite[p.~5]{devroyeProbabilisticTheoryPattern1996}
\begin{equation}\label{eq:bayes_classifier_general}
    g^*(x) = \arg\max_{y \in \mathcal{Y}} P(Y=y\mid X=x),
\end{equation}
where $P(Y=y\mid X=x)$ is the posterior probability of class $y$ given observation $x$. The Bayes risk $R^*$ is then the probability of error of this optimal classifier:
\begin{equation}\label{eq:bayes_risk_general}
    R^* = P\big(g^*(X) \neq Y\big).
\end{equation}
\end{definition}

When there are only two classes, say $\mathcal{Y}=\{0,1\}$, which is the exact case for single-qubit quantum state discrimination, the Bayes classifier in Eq.~\eqref{eq:bayes_classifier_general} takes the particularly simple form
\begin{equation*}
g^*(x) =
\begin{cases}
1 & \text{if } P(Y=1\mid X=x) \ge 1/2,\\
0 & \text{otherwise}.
\end{cases}
\end{equation*}
For equal class priors, $P(Y=0)=P(Y=1)=1/2$, the corresponding Bayes risk in Eq.~\eqref{eq:bayes_risk_general} has the form
\begin{equation}\label{eq:bayes_risk_binary}
    R^* = \frac{1}{2} - \frac{1}{4}\int|f_0(x) - f_1(x)|dx,
\end{equation}
where $f_y(x)$ is the conditional probability density of the continuous observation $X$ given $Y=y$. For any classifier $g(x)$, its error probability satisfies $R(g) \ge R^*$\cite[Theorem 2.1]{devroyeProbabilisticTheoryPattern1996}.

For $d$ readout channels, let $\mathbf{x}\in\mathbb{R}^{2d}$ collect their joint I--Q readout signals, let $\pi_y=P(Y=y)$ be the class prior, and let $f_y^{(d)}$ be the corresponding class-conditional density. The Bayes risk can also be written as
\begin{equation}
    R_{2d}^{*}
    =
    \int_{\mathbb{R}^{2d}}
    \min\{\pi_0 f_0^{(d)}(\mathbf{x}),
          \pi_1 f_1^{(d)}(\mathbf{x})\}
    \,d\mathbf{x}.
    \label{eq:bayes_risk_min_form}
\end{equation}
Because the Bayes risk is a fundamental lower bound on the error probability, it limits how well any quantum state discriminator can perform on single-qubit data, regardless of the sophistication of the training strategy or architecture. Even \textit{in situ} training at a particular measurement cycle can only move the discriminator toward the cycle-specific lower bound. It cannot reduce the error below the Bayes risk determined by the overlapping distributions. Thus, once MCM dynamics cause the analog distributions for $\ket{0}$ and $\ket{1}$ to overlap, no discriminator trained on the same single-qubit observation space can reduce its error below $R^*$.

The next section makes this limitation concrete by examining state deterioration through changes in the readout distribution of a reused qubit relative to SPAM calibration during a representative IBM superconducting-hardware MCM circuit.

%!TEX root = main.tex
\section{Phenomenology of State Deterioration in MCM Circuits}\label{sec:phenomenology}

We now examine the state deterioration that occurs in MCM circuits on IBM superconducting-qubit hardware. The background in Section~\ref{sec:background} establishes that SPAM-trained discriminators are limited by the Bayes risk of their observed single-qubit I--Q distributions. Here, we show that the relevant I--Q distributions in an MCM circuit can differ sharply from the SPAM calibration distribution used for training.

\begin{figure*}[t]
    \centering

    \begin{minipage}[t]{0.30\textwidth}
        \textbf{(a)}\\
        \vspace{-0.7cm}
                \centering
        \includegraphics[width=\linewidth]{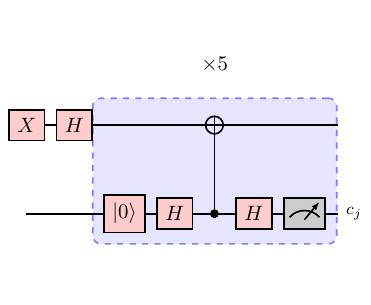}
    \end{minipage}
    \hfill
    \begin{minipage}[t]{0.31\textwidth}
        \textbf{(b)}\par
        \vspace{0.1cm} % Minor nudge to visually align the plot box with the floating circuit
        \includegraphics[width=\linewidth]{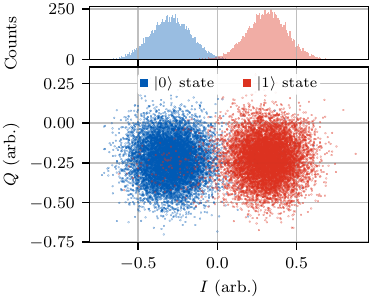}
    \end{minipage}
    \hfill
    \begin{minipage}[t]{0.31\textwidth}
        \textbf{(d)}\par
        \vspace{0.1cm}
        \includegraphics[width=\linewidth]{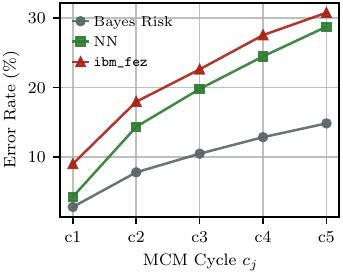}
    \end{minipage}

    \vspace{0.4cm}

    \begin{minipage}[t]{\textwidth}
        \textbf{(c)}\par
        \includegraphics[width=\linewidth]{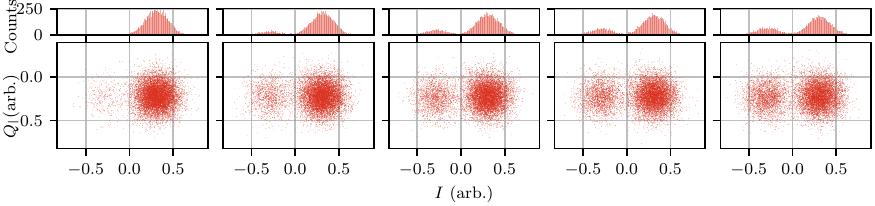}
    \end{minipage}

    \caption{State deterioration reflected in readout-distribution changes in the five-cycle QR-Stabilizer workload. (a) Core circuit with a measured data qubit and a phase-kickback ancilla. \textit{Lifting} proxy wires are omitted. (b) SPAM references for preparation labels $0$ (blue) and $1$ (red) on qubit 51 of \texttt{ibm\_fez}. 10,000 shots per class are sampled from the reference acquisitions for display. (c) Target-readout I--Q distributions ordered from cycles $c_1$ to $c_5$. (d) Equal-prior kernel density estimation (KDE) Bayes-risk estimates comparing the SPAM-$0$ reference with each ideal-$1$ target cycle, SPAM-trained NN-based QSD error rates, and transferred digital-error estimates. The curve labeled \texttt{ibm\_fez} transfers the accuracy offset between NN-based QSD and digital readout from a separate acquisition. It is not native digital readout of the same target shots.}
    \label{fig:main_readout}
\end{figure*}

The QR-Stabilizer readout distribution departs increasingly from the SPAM reference as the circuit proceeds, with the final-cycle observations spread broadly across the I--Q plane (Fig.~\ref{fig:main_readout}). We observe this trend on \texttt{ibm\_fez} using the acquisitions described in Appendix~\ref{app:workload_definitions}. This evidence of state deterioration motivates matching calibration to the measurement context.

Possible contributors to state deterioration and the observed distributional changes include relaxation and dephasing of the measured data qubit, measurement backaction, and gate errors. Residual photons in the readout resonator can also perturb subsequent operations~\cite{mcclureRapidDrivenReset2016,zhaoSingleStepPhaseEngineeredPulse2026}. Errors on the phase-kickback ancilla can also affect the data qubit through their two-qubit interaction. The present data do not separately isolate these mechanisms.

To characterize state deterioration through changes in the MCM readout signals in Fig.~\ref{fig:main_readout}(c) relative to the SPAM calibration data in Fig.~\ref{fig:main_readout}(b), we compute the Bayes risk between the distribution of the readout signals in the MCM circuit and the distribution of the SPAM calibration signals. If we denote the distribution of the MCM readout signals by $f^{c_j}_{1}$, where the ideal result for the measured data qubit is $\ket{1}$ at measurement cycle $c_j$, and the distribution of the SPAM calibration signals by $f^{\rm SPAM}_{0}$, then according to Eq.~\eqref{eq:bayes_risk_binary}, the Bayes risk is
\begin{equation*}
    R^* = \frac{1}{2} - \frac{1}{4}\int_{\mathbf{x}\in\mathbb{R}^2} |f^{c_j}_{1}(\mathbf{x}) - f^{\rm SPAM}_{0}(\mathbf{x})| d\mathbf{x}.
\end{equation*}

Figure~\ref{fig:main_readout}(d) shows increasing SPAM-trained NN-based QSD error alongside the cycle-dependent overlap estimate and the transferred digital-error estimate. The overlap calculation pairs a static $0$ reference with each ideal-$1$ target cycle, whereas the NN-based QSD error is evaluated on target-workload shots. Their separation is therefore not an estimate of the classifier's excess risk above a Bayes optimum for the same task.

The evidence of state deterioration in Fig.~\ref{fig:main_readout}(b,c) motivate \textit{in situ} calibration, whose benefit is tested directly in Section~\ref{subsec:insitu_cal}. Independently of this experiment, Bayes risk limits classification for any fixed set of observations. Adding an auxiliary readout offers a way to change those observations and potentially reduce their overlap. The gap between the plotted curves alone does not establish that this change is necessary.

%!TEX root = main.tex
\begingroup
\raggedbottom
\section{Framework}\label{sec:framework}
\textit{Lifting} combines context-matched calibration, repetition encoding, and neural decoding. Calibration adapts the decoder to the workload, while encoding supplies auxiliary observations that can distinguish outcomes overlapping in a single-qubit I--Q plane. We first describe these components, then derive an overlap bound and introduce optional confidence selection for evaluated shots.

\begin{figure*}[!t]
    \centering
    \includegraphics[width=\textwidth]{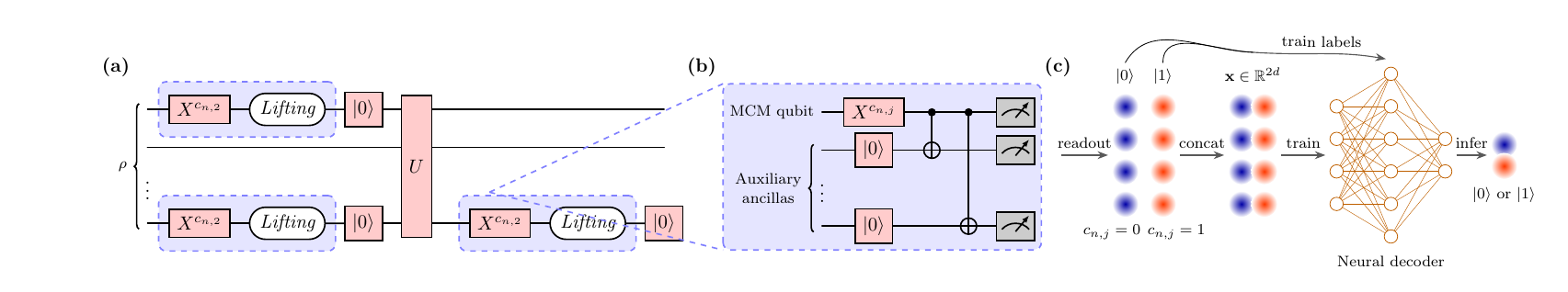}
    \caption{\textit{Lifting} workflow. (a) Context-matched calibration executions prepare reference labels and place encoded readout at the workload's measurement locations. (b) Local repetition encoding couples the labeled MCM qubit to $d-1$ auxiliary ancillas initialized to $\ket{0}$. Concurrent readout supplies their joint I--Q readout signals. (c) The $d$ I--Q pairs are concatenated into a $2d$-component feature vector. Reference labels supervise neural training, and inference assigns a logical outcome to each target observation. Optional confidence selection retains a subset of observations before reporting retained-shot accuracy.}
    \label{fig:overview}
\end{figure*}

\subsection{\textit{In Situ} Calibration for State Deterioration}\label{subsec:insitu_cal}

The QR-Stabilizer measurements in Section~\ref{sec:phenomenology} show state deterioration through changes in readout distributions within a circuit. We therefore acquire calibration data at the workload's measurement locations, including the preceding operations that shape the readout context. This procedure aims to approximate the conditions seen by target shots. Preparation differences and hardware variability can still produce a mismatch.

The calibration circuits in Fig.~\ref{fig:overview}(a) reproduce the relevant circuit depth, spectator activity, and measurement timing while preparing reference labels $y\in\{0,1\}$. For a single readout channel, they produce the labeled dataset
\begin{equation}\label{eq:insitu_cal_datasets}
    \mathcal{D}_{\mathrm{cal}} = \left\{ \left( \mathbf{x}_i, y_i \right) \right\}_{i=1}^{n},
\end{equation}
where each $\mathbf{x}_i \in \mathbb{R}^{2}$ represents a single-shot measurement vector, and $y_i \in \{0, 1\}$ is its corresponding known preparation label.

Figure~\ref{fig:insitu_calibration} compares the SPAM-trained and cycle-matched single-channel decoders on QR-Stabilizer target shots. The visible improvement is concentrated in the early measurements, while later-cycle errors remain high in this experiment. This limited benefit motivates adding readout information as well as changing calibration. The constructed risk curve provides a separate overlap diagnostic rather than a matched lower bound for these target-shot errors.

\subsection{Repetition Encoding}\label{subsec:rep_encoding}

For a given MCM circuit, \textit{Lifting} leaves the circuit unchanged except at measurement locations. Immediately before measuring a target qubit, the protocol allocates $(d-1)$ auxiliary physical qubits and forms a $(d,1)$ repetition-style encoding, as shown in Fig.~\ref{fig:overview}(b). If the target qubit is

\begin{equation*}
    \ket{\psi}=a\ket{0}+b\ket{1},
\end{equation*}
then we apply $(d-1)$ controlled-NOT (CNOT) gates from the target to the auxiliary qubits, producing
\begin{equation}
    a\ket{0}+b\ket{1}
    \longmapsto
    a\ket{0}^{\otimes d}+b\ket{1}^{\otimes d}.
    \label{eq:lifting_encoding}
\end{equation}
This encoding is the same physical resource used by active readout-error mitigation and repetition-code-based schemes \cite{hicksActiveReadouterrorMitigation2022,guntherImprovingReadoutQuantum2021,sohnQuantumErrorMitigation2026}. The difference is in the decoder. Instead of independently thresholding each qubit and applying majority voting, all $d$ qubits are measured concurrently and the joint I--Q readout vector is decoded.

The auxiliary qubits must be initialized to $\ket{0}$ before the CNOT fanout. When auxiliaries are reused within a circuit, a reset prepares them for the next encoded measurement. Imperfect initialization and added gates can reduce the benefit of the auxiliary observations, so their costs must be assessed on hardware.

The $d$ encoded qubits are then measured simultaneously using multiplexed dispersive readout. Each shot produces $d$ pairs of in-phase and quadrature values,
\begin{equation}
    \mathbf{x}^{(d)}=(I_1,Q_1,\ldots,I_d,Q_d)\in\mathbb{R}^{2d}.
    \label{eq:lifting_feature}
\end{equation}
This vector is the basic input to analog filtering and neural decoding. With the calibrated distributions and the joint I--Q readout space defined, we can now derive the limits of their distinguishability.

The repetition encoding in Fig.~\ref{fig:overview}(b) also applies to the \textit{in situ} calibration circuits in Fig.~\ref{fig:overview}(a). Replacing their ordinary measurements with encoded readout produces $2d$-dimensional joint I--Q readout vectors for training. Equation~\eqref{eq:insitu_cal_datasets} then becomes
\begin{equation}\label{eq:insitu_cal_datasets_dim_2d}
     \mathcal{D}_{\mathrm{cal}} = \left\{ \left( \mathbf{x}_i^{(d)}, y_i \right) \right\}_{i=1}^{n},
\end{equation}
where each $\mathbf{x}_i^{(d)}\in\mathbb{R}^{2d}$ is the $2d$-dimensional joint I--Q readout vector corresponding to the $i$th shot and $y_i\in\{0,1\}$ is the associated class label.

\begin{figure}[!t]
    \centering
    \includegraphics[width=0.78\columnwidth]{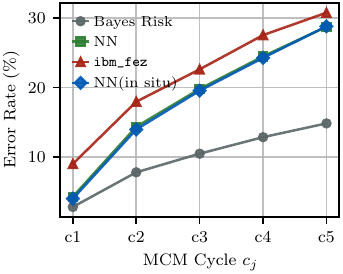}
    \caption{Error rates across five MCM cycles: the KDE Bayes-risk estimate, SPAM-trained NN-based QSD, the transferred \texttt{ibm\_fez} digital estimate from Fig.~\ref{fig:main_readout}(d), and single-channel NN-based QSD with cycle-matched \textit{in situ} calibration. The risk curve compares constructed reference--target distributions. The decoder curves evaluate target shots. Contextual calibration changes the training references, while this comparison retains the single-channel observation space.}
    \label{fig:insitu_calibration}
\end{figure}

\subsection{Upper Bounds on the Joint Bayes Risk}

An additional qubit can separate outcomes that appear indistinguishable in the original observation space. Let $f_y^{(d)}$ denote the density of the first $d$ readout channels for label $y$, and let $\rho_d=\mathrm{BC}(f_0^{(d)},f_1^{(d)})$ be their Bhattacharyya coefficient (BC), which measures distributional overlap~\cite{fukunagaIntroductionStatisticalPattern1990}. The following result separates the benefit of retaining more observations from the stronger assumption that each added channel uniformly contracts overlap.

\begin{theorem}\label{app:main}
Consider joint I--Q readout vectors consisting of the first $d$ channels of one fixed joint distribution, with equal class priors and class-conditional densities. The optimal error $R_{2d}^*$ cannot increase when another channel is retained. If the conditional overlap of each added channel $k$ is uniformly bounded by $\gamma_{\max,k}<1$, then
\begin{equation}
    R_{2d}^{*}\le\frac{1}{2}\rho_d
    \le\frac{1}{2}\rho_1\prod_{k=2}^{d}\gamma_{\max,k}.
    \label{eq:thm_geometric_bound}
\end{equation}
If all factors are bounded by a common $q<1$, this upper bound decreases geometrically as $\frac{1}{2}\rho_1q^{d-1}$.
\end{theorem}

\begin{figure}
    \centering
    \includegraphics[width=\columnwidth]{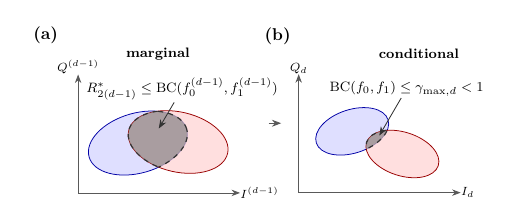}
    \caption{Qualitative geometry of marginal and conditional overlap. (a) The preceding $d-1$ readout channels have marginal Bhattacharyya overlap $\rho_{d-1}=\mathrm{BC}(f_0^{(d-1)},f_1^{(d-1)})$, which upper-bounds balanced-prior Bayes risk after multiplication by $1/2$. (b) Conditional overlap in the added I--Q channel is bounded by $\gamma_{\max,d}$. Shaded regions schematically illustrate the factorization.}
    \label{fig:geometric_bayes_bound}
\end{figure}

To see the geometry, divide the observation into the preceding channels $\mathbf{u}=\mathbf{x}^{(d-1)}$ and the added channel $\mathbf{z}=\mathbf{x}_d$. Their joint density factors as
\begin{equation}\label{eq:lifting_dist_factor}
    f_y^{(d)}(\mathbf{u},\mathbf{z})
    =f_y^{(d-1)}(\mathbf{u})f_y(\mathbf{z}\mid\mathbf{u}).
\end{equation}
At a fixed location $\mathbf{u}$, the added channel provides another view of the two classes (Fig.~\ref{fig:geometric_bayes_bound}). Its conditional overlap $\gamma_d(\mathbf{u})$ weights the overlap remaining in the preceding channels, giving
\begin{equation*}
\begin{split}
    \rho_d
    &=\int\sqrt{f_0^{(d-1)}(\mathbf{u})f_1^{(d-1)}(\mathbf{u})}
       \,\gamma_d(\mathbf{u})\,d\mathbf{u}\\
    &\le\gamma_{\max,d}\rho_{d-1}.
\end{split}
\end{equation*}
Repeating this step gives the product bound. The exact optimal error is non-increasing because a decoder can always ignore the new channel. Strict improvement requires the added observations to split a region previously assigned one label into regions favoring different classes, with positive probability. Appendix~\ref{app:proof} gives the precise conditions and the Gaussian special case.

The theorem compares coordinates of a fixed distribution. Physically increasing the encoding distance also changes gates and measurement conditions, so separately executed distance sweeps need not share the required marginals. The hardware evaluation therefore examines the benefit of joint I--Q readout signals together with the costs of implementing them.

\subsection{Decoding}

\subsubsection{Analog Confidence Filtering}

Confidence filtering selects evaluated shots whose decoded outcomes are sufficiently well supported by their analog observations. It is optional and trades retained-shot coverage for accuracy. It is neither an assumption of Theorem~\ref{app:main} nor a modification of the reference labels used for training. The QR-Stabilizer implementation uses a Gaussian confidence score. The quantum error correction (QEC) workload considered here consists of five parity-check rounds and is denoted QEC-5. This workload and teleportation use neural posterior confidence, as specified in Appendix~\ref{app:workload_decoders}.

Let $\mathcal{D}^0_{\mathrm{cal}}$ and $\mathcal{D}^1_{\mathrm{cal}}$ be the subsets of the \textit{in situ} calibration dataset in Eq.~\eqref{eq:insitu_cal_datasets_dim_2d} containing observations labeled $0$ and $1$, respectively. For QR-Stabilizer, using a multivariate Gaussian approximation for each reference class, we can estimate the mean and covariance of the joint I--Q readout vector:
\begin{align}
    \hat{\boldsymbol{\mu}}_y
    &=
    \frac{1}{|\mathcal{D}^y_{\mathrm{cal}}|}
    \sum_{\mathbf{x}_i\in\mathcal{D}^y_{\mathrm{cal}}}\mathbf{x}_i,\\
    \hat{\boldsymbol{\Sigma}}_y
    &=
    \frac{1}{|\mathcal{D}^y_{\mathrm{cal}}|-1}
    \sum_{\mathbf{x}_i\in\mathcal{D}^y_{\mathrm{cal}}}
    (\mathbf{x}_i-\hat{\boldsymbol{\mu}}_y)
    (\mathbf{x}_i-\hat{\boldsymbol{\mu}}_y)^{T}.
\end{align}
With a small diagonal regularizer $\lambda I$, define the Gaussian score
\begin{equation}\label{eq:decision_score}
\begin{split}
    s_y(\mathbf{x})
    ={}&
    -\frac{1}{2}
    (\mathbf{x}-\hat{\boldsymbol{\mu}}_y)^T
    (\hat{\boldsymbol{\Sigma}}_y+\lambda I)^{-1}
    (\mathbf{x}-\hat{\boldsymbol{\mu}}_y)\\
    &-\frac{1}{2}\log\det(\hat{\boldsymbol{\Sigma}}_y+\lambda I)
    +\log\hat{\pi}_y,
\end{split}
\end{equation}
where $\hat{\pi}_y=|\mathcal{D}^y_{\mathrm{cal}}|/|\mathcal{D}_{\mathrm{cal}}|$ is the empirical class prior. The posterior confidence is
\begin{equation}
    c(\mathbf{x})
    =
    \max_{y\in\{0,1\}}
    \frac{\exp(s_y(\mathbf{x}))}
    {\exp(s_0(\mathbf{x}))+\exp(s_1(\mathbf{x}))}.
    \label{eq:analog_confidence}
\end{equation}
An evaluated QR-Stabilizer shot is retained if $c(\mathbf{x})\ge \tau$ at every measurement cycle, for a chosen threshold $\tau$. Equivalently, one can threshold the absolute log-likelihood ratio $|s_1(\mathbf{x})-s_0(\mathbf{x})|$. This formulation is the analog counterpart of digital detection or postselection: it rejects outcomes whose calibration-derived likelihood is too close to the decision boundary. Unlike majority voting, however, the confidence is computed in the full $2d$-dimensional I--Q feature space.

\subsubsection{Neural Decoding}
The overlap analysis describes the optimal classification limit for the chosen observations. A neural decoder provides a flexible approximation to the corresponding decision rule. The following existence result, proved in Claim~\ref{app:universal_approximation} of Appendix~\ref{app:proof}, motivates this choice without guaranteeing that finite training reaches the optimum.

\begin{proposition}
    For any $\epsilon>0$, there exists a neural network $f_{\boldsymbol{\theta}} \in \mathcal{H}_{2d}$ such that the misclassification risk $R(g_\theta)=P(g_\theta(X)\ne Y)$ of the thresholded classifier $g_\theta(\mathbf{x})=\mathbf{1}\{f_{\boldsymbol{\theta}}(\mathbf{x})\ge1/2\}$ is bounded by $R^*_{2d} + \epsilon$, where $R^*_{2d}$ is the Bayes risk in Theorem~\ref{app:main} and $\mathcal{H}_{2d}$ is a sufficiently expressive family of continuous feedforward neural networks operating on the $2d$-dimensional joint I--Q readout space.
 \end{proposition}

Using the labeled calibration observations, we train a neural decoder
\begin{equation*}
    f_{\boldsymbol{\theta}}:\mathbb{R}^{2d}\rightarrow[0,1],
\end{equation*}
where $f_{\boldsymbol{\theta}}(\mathbf{x})$ estimates the probability that the logical outcome is $1$. We use a multilayer perceptron (MLP), consistent with prior NN-based QSD methods and FPGA-oriented single-shot discrimination systems \cite{lienhardDeepNeuralNetworkDiscriminationMultiplexed2022,mauryaScalingQubitReadout2023,voraMLPoweredFPGAbasedRealTime2024}. The input layer has $2d$ nodes, one for each I--Q component in Eq.~\eqref{eq:lifting_feature}. The hidden layers use rectified linear activations, and the output layer uses a sigmoid activation for binary logical readout. The MCM experiments use one hidden layer of 32 neurons. The output heads and training settings are specified in Appendix~\ref{app:workload_decoders}.

The decoder is trained on the \textit{in situ} calibration set by minimizing binary cross entropy, as illustrated in Fig.~\ref{fig:overview}(c),
\begin{equation*}
    \mathcal{L}(\boldsymbol{\theta})
    =
    -\frac{1}{n}
    \sum_{i=1}^{n}
    \left[
    y_i\log f_{\boldsymbol{\theta}}(\mathbf{x}_i)
    +(1-y_i)\log(1-f_{\boldsymbol{\theta}}(\mathbf{x}_i))
    \right].
\end{equation*}
During inference, each encoded measurement produces a joint I--Q readout vector $\mathbf{x}\in\mathbb{R}^{2d}$. The logical single-shot output is
\begin{equation*}
    \hat{y}
    =
    \begin{cases}
    1, & f_{\boldsymbol{\theta}}(\mathbf{x})\ge 1/2,\\
    0, & f_{\boldsymbol{\theta}}(\mathbf{x})<1/2.
    \end{cases}
\end{equation*}
The decoder therefore replaces the conventional pipeline of independent qubit discrimination followed by majority voting with a direct learned decision boundary in the joint I--Q readout space.
\endgroup

%!TEX root = main.tex
\section{Experimental Results}\label{sec:results}

\subsection{Hardware Readout Geometry and Empirical Accuracy Model}

We examine the benefit of joint I--Q readout signals on IBM hardware in two steps. First, we compare the target-only and joint I--Q readout distributions from a distance-two acquisition. We then use a fitted empirical model to describe how assignment accuracy changes with the physical encoding distance, where additional gates and readout channels can introduce competing costs.

Figure~\ref{fig:bayes_hardware_val} shows less overlap in the quadratic-discriminant scores computed from joint I--Q readout signals than in the target-only scores. The equal-prior KDE risk estimate falls from 14.7\% for the target channel to 5.9\% for the joint readout. This supports the benefit of retaining auxiliary information from the same acquisition.

For the distance sweep, we use a phenomenological product model inspired by the overlap analysis. Let $\epsilon_{Y,k}$ be the empirical marginal assignment-error input for preparation label $Y$, with $k=0$ denoting the target and $k=1,\ldots,d-1$ its auxiliaries. Define the fitted auxiliary factors
\begin{equation}
    g_{Y,k}=(2\epsilon_{Y,k})^\kappa,
    \qquad 0<\kappa<1,
    \label{eq:kappa_synergy}
\end{equation}
and the modeled class-specific error
\begin{equation}
    \widehat P_{\mathrm{err}}(Y)
    =\epsilon_{Y,0}\prod_{k=1}^{d-1}g_{Y,k}.
    \label{eq:readout_err}
\end{equation}
The factors $g_{Y,k}$ summarize how the marginal error inputs combine in the fit. They are empirical model factors, rather than conditional Bhattacharyya overlaps. Thus, this fitted model and the measured risk reduction illustrate the value of auxiliary signals without establishing the conditional-overlap assumptions of Theorem~\ref{app:main}.

We include an asymmetric attenuation factor to represent the implementation cost of the encoding. With $e_g^{(j)}$ denoting the two-qubit gate-error input for encoding interaction $j$, the balanced-class accuracy model is
\begin{equation}
\begin{split}
A_{\mathrm{fit}}=\frac12\Bigg[&1-\widehat P_{\mathrm{err}}(0)\\
&+\prod_{j=1}^{d-1}(1-\alpha e_g^{(j)})
 \bigl(1-\widehat P_{\mathrm{err}}(1)\bigr)\Bigg].
\label{eq:final_ann}
\end{split}
\end{equation}
This model uses $0\le\epsilon_{Y,k}\le1/2$ and $0\le\alpha e_g^{(j)}\le1$, so the factors and modeled accuracy remain in $[0,1]$. The extra attenuation of the class-one branch represents an asymmetric implementation cost without assigning it to a specific noise mechanism.

\begin{figure}[!htbp]
    \centering
    \includegraphics[width=0.47\textwidth]{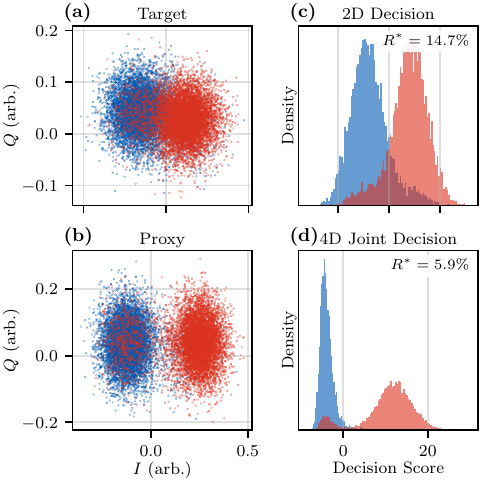}
    \caption{Target-only and joint I--Q readout discrimination for $d=2$ on \texttt{ibm\_fez}, using 10,000 shots per preparation label. (a,b) Target and auxiliary (proxy) I--Q scatter, with blue and red denoting reference labels $0$ and $1$, respectively. (c,d) Signed quadratic-discriminant scores $s_1-s_0$ from the two-dimensional target observations and four-dimensional joint I--Q readout signals. Annotated risks are equal-prior KDE-based estimates of distributional overlap, rather than held-out classifier error rates.}
    \label{fig:bayes_hardware_val}
\end{figure}

Figure~\ref{fig:simulation_results} shows that the fitted model follows the hardware accuracy trend over the tested Fez distances. The parameters $\kappa=0.45$ and $\alpha=9.03$ are fitted to this sweep. The trend is consistent with a balance between useful auxiliary information and the errors introduced by implementing a larger encoding.

\begin{figure}[!t]
\centering
\includegraphics[width=0.45\textwidth]{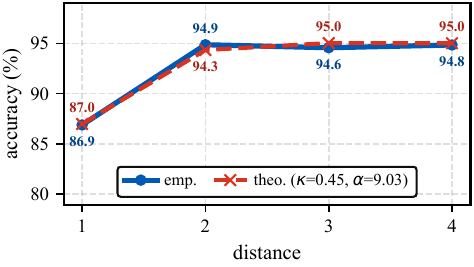}
\caption{IBM Fez assignment accuracy versus encoding distance $d$, comparing hardware measurements (emp.) with the fitted empirical model in Eq.~\eqref{eq:final_ann}. The curve labeled theo. uses $\kappa=0.45$ and $\alpha=9.03$, fitted for the displayed distance sweep.}
\label{fig:simulation_results}
\end{figure}

\subsection{Comparison with Repetition and Temporal Schemes}

We first evaluate how the repetition encoding scales with the code distance $d$ for spatial readout and sequential temporal acquisition. Figure~\ref{fig:exp1_lines_grid} compares the baseline readout, a distance-$d$ repetition code, a temporal repeated-measurement scheme, and our encoding on three IBM superconducting processors. The experiments use neighboring qubits on each device, selecting a target qubit and then adding nearby auxiliary qubits as $d$ increases. The baseline reports the average native assignment accuracy of the selected physical qubits. The spatial repetition comparator independently discriminates the concurrent readouts and applies a hard vote. The temporal comparator applies neural decoding to $d$ sequential observations of one qubit. \textit{Lifting} applies neural decoding to the joint I--Q readout signals from the $d$ encoded qubits in Section~\ref{subsec:rep_encoding}. These schemes increase physical-qubit count or measurement time differently as $d$ grows.

\begin{figure*}[t!]
\centering
\subfigure[\texttt{ibm\_fez}.]{
\centering
\includegraphics[width=2.2in]{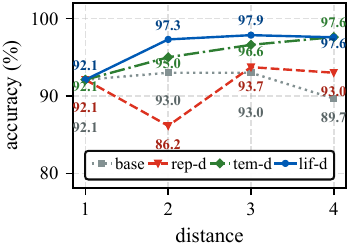}
\label{fig:exp2_ibm_fez_lines}
}
\hspace{-0.1in}
\subfigure[\texttt{ibm\_kingston}.]{
\centering
\includegraphics[width=2.2in]{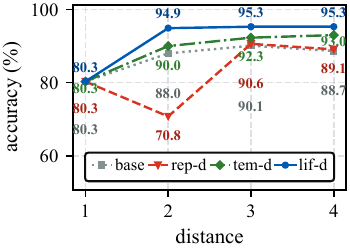}
\label{fig:exp2_ibm_kingston_lines}
}
\hspace{-0.1in}
\subfigure[\texttt{ibm\_marrakesh}.]{
\centering
\includegraphics[width=2.2in]{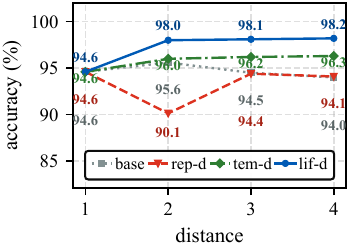}
\label{fig:exp2_ibm_marrakesh_lines}
}
\caption{Assignment accuracy versus distance on (a) \texttt{ibm\_fez}, (b) \texttt{ibm\_kingston}, and (c) \texttt{ibm\_marrakesh}. The baseline averages native assignment accuracy over the selected physical qubits. Spatial repetition (rep-$d$) uses independent hard assignments followed by voting. The temporal comparator (tem-$d$) jointly decodes $d$ sequential I--Q observations of one qubit. \textit{Lifting} (lif-$d$) jointly decodes concurrent I--Q observations from $d$ repetition-encoded qubits. The spatial schemes increase physical-qubit count with $d$, whereas the temporal scheme increases measurement time. Both neural schemes use $2d$ scalar inputs.}
\label{fig:exp1_lines_grid}
\end{figure*}

\begin{figure*}[t!]
\centering
\subfigure[\texttt{ibm\_fez}.]{
\centering
\includegraphics[width=2.2in]{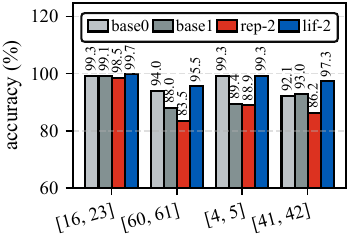}
\label{fig:exp1_ibm_fez_bars}
}
\hspace{-0.1in}
\subfigure[\texttt{ibm\_kingston}.]{
\centering
\includegraphics[width=2.2in]{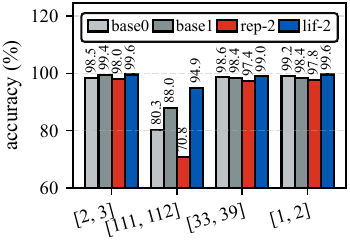}
\label{fig:exp1_ibm_kingston_bars}
}
\hspace{-0.1in}
\subfigure[\texttt{ibm\_marrakesh}.]{
\centering
\includegraphics[width=2.2in]{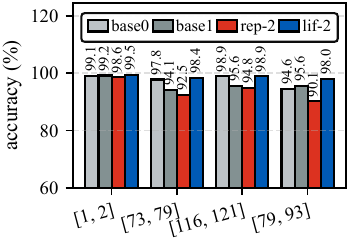}
\label{fig:exp1_ibm_marrakesh_bars}
}
\caption{Distance-two assignment accuracy for physical-qubit pairs on (a) \texttt{ibm\_fez}, (b) \texttt{ibm\_kingston}, and (c) \texttt{ibm\_marrakesh}. Bracketed labels identify the selected qubit pairs. The base0 and base1 bars report native assignment accuracy of the first and second qubits, respectively. The rep-2 bars use hard repetition decoding. \textit{Lifting}, labeled lif-2, uses neural decoding of joint I--Q readout signals from the same two-qubit encoding. Pair selection spans heterogeneous native readout accuracies.}
\label{fig:exp1_bars_grid}
\end{figure*}

The hard-vote comparator benefits unevenly from increasing distance, with a pronounced loss at $d=2$ and no consistent gain from $d=3$ to $d=4$. At even distances, conflicting assignments can produce a tie, so the result depends on how that tie is resolved. Additional gate and readout errors can further offset the benefit of another qubit.

\textit{Lifting} generally has higher assignment accuracy than hard repetition decoding, while the curves coincide at $d=1$. At distance two, the relative error reductions over distance-two hard repetition decoding are 80.4\%, 82.5\%, and 79.8\% on \texttt{ibm\_fez}, \texttt{ibm\_kingston}, and \texttt{ibm\_marrakesh}, respectively. Neural decoding of joint I--Q readout signals can use confidence information lost by independent hard assignments, and these results support distance two as a practical choice with one auxiliary qubit. Larger distances provide smaller or device-dependent gains.

\textit{Lifting} also generally exceeds the temporal comparator, with coincident starting points and a near tie on Fez at the largest distance. Although both neural decoders receive $2d$ scalar features, spatial encoding uses more qubits and entangling gates, whereas temporal acquisition uses repeated measurements and more time. Their different circuits and acquisition conditions make this a comparison of implemented schemes, rather than an isolated test of feature dimension or a single noise mechanism.

\subsection{Distance-Two Accuracy and Auxiliary-Qubit Cost}

\begin{figure}[!ht]
    \centering
    \includegraphics[width=0.9\columnwidth]{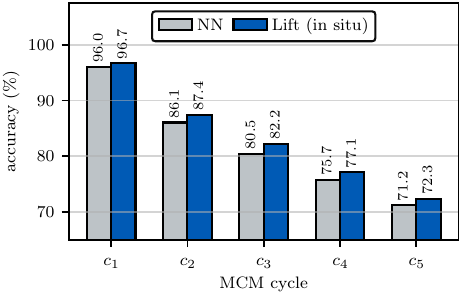}
    \caption{Per-cycle QR-Stabilizer discrimination accuracy on \texttt{ibm\_fez}, in measurement order $c_1$--$c_5$. NN-based QSD uses target shots. \textit{Lifting} with \textit{in situ} calibration uses held-out calibration shots (Appendix~\ref{app:workload_decoders}).}
    \label{fig:qr_stab_acc_per_meas}
\end{figure}

\begin{figure*}[t!]
    \centering
    \includegraphics[width=\textwidth]{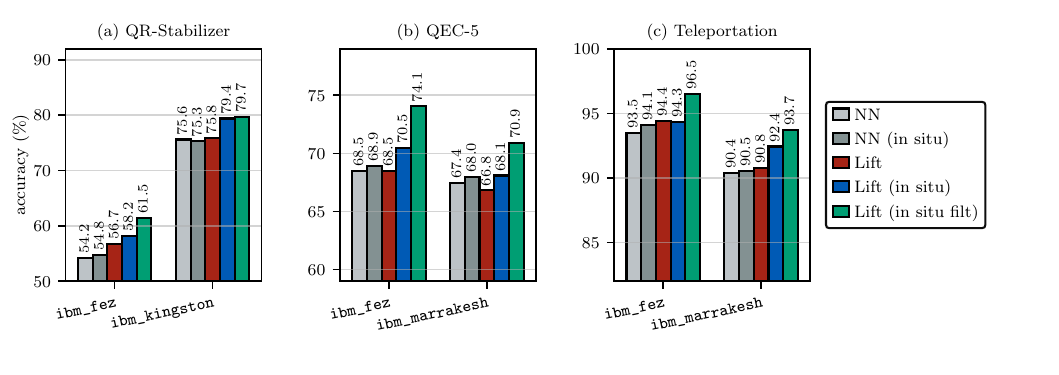}
    \caption{Circuit success rates for (a) QR-Stabilizer, (b) QEC-5, and (c) teleportation. Success requires a correct complete QR-Stabilizer trace, final QEC Pauli-frame bit, or recovered computational-basis input, respectively, as defined in Appendix~\ref{app:workload_definitions}. Green bars report success among retained shots. Decoder configurations and filtering rules are given in Appendix~\ref{app:workload_decoders}.}
    \label{fig:combined_workload_results}
\end{figure*}

Figure~\ref{fig:exp1_bars_grid} examines the selected distance-two operating point across qubit pairs with different native assignment accuracies. The bars \textit{base0} and \textit{base1} identify the individual-qubit baselines, while hard repetition decoding and \textit{Lifting} process the same two-qubit encoding. The pairs include both similar-quality qubits and cases with a substantial difference between their native accuracies.

\textit{Lifting} exceeds hard repetition decoding for every displayed pair, including pairs with unequal native accuracies. Against the individual-qubit baselines, \textit{Lifting} is higher in many cases and approximately ties the stronger qubit in some high-accuracy pairs. The consistent gain over hard repetition supports using analog confidence to resolve ambiguous two-qubit records. With one auxiliary qubit, distance two provides a practical accuracy--resource tradeoff for the subsequent MCM experiments.

\subsection{Performance on MCM Workloads}

\begin{samepage}
We next evaluate distance-two \textit{Lifting} on QR-Stabilizer, QEC-5, and teleportation, using the circuits and reference labels defined in Appendix~\ref{app:workload_definitions}. These workloads test repeated measurement and reuse, parity-based error tracking, and recovery of a teleported computational-basis input, respectively. We distinguish the accuracy of each individual measurement from the success rate of the circuit under its workload-specific criterion.\par
\end{samepage}

Figure~\ref{fig:qr_stab_acc_per_meas} shows the per-cycle discrimination accuracy for each measurement in QR-Stabilizer: both series decline as the circuit proceeds, while \textit{Lifting} with \textit{in situ} calibration remains higher throughout. The downward trend is consistent with the state deterioration identified in Sec.~\ref{sec:phenomenology}, which makes later measurements harder to classify. Because NN-based QSD is evaluated on target shots and \textit{Lifting} on held-out calibration shots, their separation does not establish a paired improvement on the same workload shots.

Figure~\ref{fig:combined_workload_results} instead reports circuit success rates under the criteria in Appendix~\ref{app:workload_definitions}: all QR-Stabilizer outcomes must be correct, QEC-5 must identify the final correction frame, and teleportation must recover the prepared input bit. For QEC-5 and teleportation, corrections are tracked in software rather than applied through physical conditional gates. These criteria test the expected circuit behavior through decoded measurement records. They do not measure corrected-state or arbitrary-state teleportation fidelity.

For QR-Stabilizer, \textit{Lifting} with \textit{in situ} calibration has the highest unfiltered success rate on both devices, whereas contextual calibration alone changes single-channel NN-based QSD little and repetition encoding alone gives a smaller, device-dependent gain. On \texttt{ibm\_fez}, combining encoding and contextual calibration raises the success rate from 54.2\% to 58.2\%. For these QR-Stabilizer experiments, the combined configuration improves success more than single-channel contextual retraining alone.

For QEC-5, \textit{Lifting} with \textit{in situ} calibration also achieves the highest unfiltered success rate on both devices, although its advantage over NN-based QSD with \textit{in situ} calibration is small on \texttt{ibm\_marrakesh}. \textit{Lifting} with static calibration is comparable to SPAM-trained NN-based QSD on \texttt{ibm\_fez} and lower on \texttt{ibm\_marrakesh}. Thus, matching the decoder to the repeated parity-check circuit is important to realizing the benefit of the extra readout channel in these results.

For teleportation, all \textit{Lifting} variants outperform SPAM-trained NN-based QSD on both devices, but the ordering of the unfiltered variants depends on the device. \textit{Lifting} with only Alice's proxy and \textit{Lifting} with both Alice's and Bob's proxies are nearly equal on \texttt{ibm\_fez}, whereas the latter is higher on \texttt{ibm\_marrakesh}. Both variants use \textit{in situ} calibration, so this comparison indicates a device-dependent benefit from adding Bob's proxy readout.

Confidence filtering gives the highest retained-shot success rate in every workload and device shown. It removes ambiguous records, making the accepted outcomes more reliable at the cost of discarding shots. These green bars therefore describe a tradeoff between accuracy and retained-shot coverage, with the workload-specific selection rules given in Appendix~\ref{app:workload_decoders}.

%!TEX root = main.tex
\section{Discussion}\label{sec:discussion}
% abstract 
% smooth, professional prose and confident tone
% never semicolon

% 1. we propose Lifting, a quantum state disrimination framework for MCMs on superconducting NISQ machines.
% 2. evaluated across 3 IBM processors, achieved up to XX% improvement in representative MCM workloads. (pick highest improvement)
% 3. significance to future directions (name some) in NISQ QC. 

We have presented \textit{Lifting}, a quantum state discrimination framework for mid-circuit measurements on noisy superconducting processors. Repetition encoding supplies joint I--Q readout signals that allow the neural decoder to distinguish outcomes that overlap in a single-qubit readout plane. Complementary \textit{in situ} calibration addresses the mismatch associated with state deterioration during circuit execution. The theoretical analysis establishes the benefit of retaining additional readout channels under a fixed joint distribution and identifies sufficient conditions for contraction of an overlap-based upper bound on Bayes risk.

Experiments across three IBM processors demonstrate improved circuit success in representative MCM workloads, with distance two providing a practical balance between accuracy and auxiliary-qubit cost. The largest relative gain is approximately 13.5\% in repeated stabilizer measurement, where confidence-filtered \textit{Lifting} raises circuit success from the statically calibrated NN-based QSD baseline of 54.2\% over all evaluated shots to 61.5\% among retained shots.  Across the workloads, the results support matching both the auxiliary readout channels and calibration context to the circuit.

%\input{src/acknowledgments}
%!TEX root = main.tex
\appendix

\section{Proof of Theorem~\ref{app:main}}\label{app:proof}

We first establish monotonicity of the exact Bayes risk, then derive the conditional-overlap bound. The argument applies to general class-conditional densities. Gaussian models provide a useful geometric special case~\cite{devroyeProbabilisticTheoryPattern1996,fukunagaIntroductionStatisticalPattern1990,bishopPatternRecognitionMachine2006}.

\subsection{Joint I--Q Readout Signals and Exact Risk}

Fix a joint I--Q readout vector $\mathbf{X}^{(D)}\in\mathbb{R}^{2D}$ and a binary label $Y$. For $d\le D$, let $f_y^{(d)}$ be the density of its first $d$ channels conditioned on $Y=y$, with priors $\pi_y=P(Y=y)$. All lower-dimensional densities below are marginals of this same joint law. The Bayes risk is
\begin{equation}\label{eq:bayes_risk}
R_{2d}^*=\int\min\{\pi_0f_0^{(d)}(\mathbf{x}),
                         \pi_1f_1^{(d)}(\mathbf{x})\}\,d\mathbf{x}.
\end{equation}

Write $\mathbf{u}=\mathbf{x}^{(d-1)}$ and $\mathbf{z}=\mathbf{x}_d$. Since the integral of a pointwise minimum is no larger than the minimum of the two integrals,
\begin{equation*}
\begin{split}
R_{2d}^*
&=\int\!\int\min\{\pi_0f_0^{(d)}(\mathbf{u},\mathbf{z}),
                     \pi_1f_1^{(d)}(\mathbf{u},\mathbf{z})\}\,d\mathbf{z}\,d\mathbf{u}\\
&\le\int\min\{\pi_0f_0^{(d-1)}(\mathbf{u}),
               \pi_1f_1^{(d-1)}(\mathbf{u})\}\,d\mathbf{u}\\
&=R_{2(d-1)}^*.
\end{split}
\end{equation*}
Geometrically, retaining the added coordinate allows a finer decision boundary. The original classifier remains available by ignoring that coordinate, so optimal performance cannot worsen.

Strict improvement requires more than distinct conditional distributions. To state the condition precisely, define
\begin{equation*}
h(\mathbf{u},\mathbf{z})=
\pi_1f_1^{(d)}(\mathbf{u},\mathbf{z})-
\pi_0f_0^{(d)}(\mathbf{u},\mathbf{z}).
\end{equation*}
The reduction in risk is
\begin{equation}\label{eq:strict_risk_condition}
\begin{split}
R_{2(d-1)}^*-R_{2d}^*
=\frac12\int\Bigg[&\int|h(\mathbf{u},\mathbf{z})|\,d\mathbf{z}\\
&-\left|\int h(\mathbf{u},\mathbf{z})\,d\mathbf{z}\right|\Bigg]d\mathbf{u}.
\end{split}
\end{equation}
This quantity is positive precisely when, on a set of $\mathbf{u}$ of positive measure, $h(\mathbf{u},\cdot)$ has both positive and negative parts with nonzero integrals. In that case, the added observation resolves a single marginal decision into regions favoring different classes. Conditional distributions can differ without crossing this decision threshold, leaving the exact risk unchanged.

\subsection{Conditional Overlap and the Product Bound}

Define the Bhattacharyya coefficient~\cite{fukunagaIntroductionStatisticalPattern1990}
\begin{equation}\label{eq:rho}
\rho_d=\int\sqrt{f_0^{(d)}(\mathbf{x})f_1^{(d)}(\mathbf{x})}\,d\mathbf{x}.
\end{equation}
The inequality $\min(a,b)\le\sqrt{ab}$ gives
\begin{equation*}
R_{2d}^*\le\sqrt{\pi_0\pi_1}\,\rho_d.
\end{equation*}
For the balanced priors in Theorem~\ref{app:main}, the prefactor is $1/2$. Equality is allowed, including when the two class distributions are disjoint and both sides vanish.

On the shared marginal support, factor the densities as
\begin{equation*}
f_y^{(d)}(\mathbf{u},\mathbf{z})
=f_y^{(d-1)}(\mathbf{u})f_y(\mathbf{z}\mid\mathbf{u}).
\end{equation*}
The conditional coefficient is
\begin{equation}\label{eq:gamma_factor}
\gamma_d(\mathbf{u})=
\int\sqrt{f_0(\mathbf{z}\mid\mathbf{u})
           f_1(\mathbf{z}\mid\mathbf{u})}\,d\mathbf{z}.
\end{equation}
By Cauchy--Schwarz, $0\le\gamma_d(\mathbf{u})\le1$, with equality at one exactly when the two conditional densities agree almost everywhere. Substitution into Eq.~\eqref{eq:rho} yields
\begin{equation}\label{eq:rho_factor}
\rho_d=\int\sqrt{f_0^{(d-1)}(\mathbf{u})f_1^{(d-1)}(\mathbf{u})}
             \,\gamma_d(\mathbf{u})\,d\mathbf{u}.
\end{equation}
Thus, only locations where the marginal distributions overlap contribute to the contraction.

Let $\nu_{d-1}$ be the finite measure with density
$\sqrt{f_0^{(d-1)}f_1^{(d-1)}}$ relative to $d\mathbf{u}$. The uniform contraction assumption means
\begin{equation*}
\gamma_d(\mathbf{u})\le\gamma_{\max,d}<1
\quad\text{for }\nu_{d-1}\text{-almost every }\mathbf{u}.
\end{equation*}
Equivalently, the essential supremum of $\gamma_d$ under this overlap measure is below one. If $\nu_{d-1}$ is zero, the preceding observations already separate the classes and the risk is zero. Otherwise, the assumption is additional to pointwise separation: values below one at every location may still approach one arbitrarily closely.

Applying this bound in Eq.~\eqref{eq:rho_factor} gives
\begin{equation}\label{eq:recurrence}
\rho_d\le\gamma_{\max,d}\rho_{d-1}.
\end{equation}
Iteration and the balanced-prior risk bound establish
\begin{equation}\label{eq:geometric_bound}
R_{2d}^*\le\frac12\rho_d
\le\frac12\rho_1\prod_{k=2}^d\gamma_{\max,k},
\end{equation}
proving Eq.~\eqref{eq:thm_geometric_bound}. A common $q\in(0,1)$ with $\gamma_{\max,k}\le q$ gives geometric decay of this upper bound. Individual factors below one do not by themselves ensure a common geometric rate, and contraction of the upper bound does not establish strict reduction of the exact risk.

\subsection{Gaussian Geometry and Physical Interpretation}

For nonsingular Gaussian laws, the Bhattacharyya distance is~\cite{fukunagaIntroductionStatisticalPattern1990}
\begin{equation}\label{eq:bhattacharyya_dist}
\begin{split}
D_B={}&\frac18(\boldsymbol\mu_1-\boldsymbol\mu_0)^T
              \boldsymbol\Sigma^{-1}(\boldsymbol\mu_1-\boldsymbol\mu_0)\\
&+\frac12\log\frac{\det\boldsymbol\Sigma}
 {\sqrt{\det\boldsymbol\Sigma_0\det\boldsymbol\Sigma_1}},
\qquad
\boldsymbol\Sigma=\frac{\boldsymbol\Sigma_0+\boldsymbol\Sigma_1}{2}.
\end{split}
\end{equation}
Their coefficient is $\exp(-D_B)$. Both terms are nonnegative: the first is a positive-definite quadratic form, and the second follows from concavity of the log determinant. Hence $D_B=0$ exactly when both means and both covariance matrices coincide. Separation may arise through means, covariances, or both.

If each class-conditional joint law of $(\mathbf{u},\mathbf{z})$ is Gaussian, its conditional law of $\mathbf{z}$ given $\mathbf{u}$ is Gaussian~\cite[Sec.~2.3.1]{bishopPatternRecognitionMachine2006}. Equation~\eqref{eq:bhattacharyya_dist}, evaluated using the conditional means and covariances, then gives $\gamma_d(\mathbf{u})=\exp[-D_{B,d}(\mathbf{u})]$. A uniform lower bound $D_{B,d}(\mathbf{u})\ge\delta_d>0$ on the overlap support is sufficient for contraction with $\gamma_{\max,d}=e^{-\delta_d}$. Pointwise positivity alone is insufficient.

A Gaussian marginal does not imply a Gaussian conditional law. Moreover, conditioning a Gaussian mixture generally produces another mixture, so the single-Gaussian formula cannot be applied directly to it. The density-based proof above applies to such mixtures without this restriction.

The requirement that readout vectors share one fixed joint distribution is essential when interpreting hardware results. Retaining an auxiliary channel from the same acquisition adds information without changing the marginal target distribution. Executing a larger encoding circuit can instead change that distribution through additional gates, initialization, and measurement activity. Comparisons across physical encoding distances must therefore assess these implementation costs separately from the theorem.

\begin{claim}\label{app:universal_approximation}
Let $\mu$ be the input probability distribution on $\mathbb{R}^{2d}$, and let $\mathcal{H}_{2d}$ be a family of continuous network scores dense in $L^1(\mu)$. For any $\epsilon>0$, there is an $h_\theta\in\mathcal{H}_{2d}$ whose thresholded classifier $g_\theta(\mathbf{x})=\mathbf{1}\{h_\theta(\mathbf{x})\ge1/2\}$ has risk $R(g_\theta)=P(g_\theta(X)\ne Y)\le R_{2d}^*+\epsilon$.
\end{claim}
\begin{proof}
Choose a measurable Bayes decision $g^*(\mathbf{x})=\mathbf{1}\{\pi_1 f_1(\mathbf{x})\ge\pi_0 f_0(\mathbf{x})\}$. The input distribution is a finite measure because its total mass is one. A standard sufficient approximation family uses arbitrarily many hidden units with a bounded, continuous, nonconstant activation, such as a sigmoid~\cite[Theorems~1, 2]{hornikApproximationCapabilitiesMultilayer1991}. By $L^1(\mu)$ density, choose $h_\theta$ with $\int|h_\theta-g^*|\,d\mu<\epsilon/2$. Whenever the thresholded decisions differ, $|h_\theta-g^*|\ge1/2$. Hence
\begin{equation*}
\begin{split}
R(g_\theta)-R_{2d}^*
&\le P(g_\theta(X)\ne g^*(X))\\
&\le2\int|h_\theta-g^*|\,d\mu<\epsilon.
\end{split}
\end{equation*}
This is an existence result for an expressive family, not a finite-training guarantee.
\end{proof}

\section{Experimental Setup and Device Characteristics}\label{app:experimental_setup}

\subsection{Devices, Circuits, and Reference Labels}\label{app:workload_definitions}

The experiments use \texttt{ibm\_fez}, \texttt{ibm\_kingston}, and \texttt{ibm\_marrakesh}, with the device coverage specified in each figure. Preparation labels identify reference circuits. They are not independent verification of each shot's instantaneous physical state. Static SPAM references prepare $0/1$ for single-channel readout or $00/11$ for joint I--Q readout. Contextual references include preceding workload operations and are matched to the measurement cycle.

QR-Stabilizer (Fig.~\ref{fig:mcm_workloads}(a)) uses data qubit 51, phase-kickback ancilla 52, and, for \textit{Lifting}, proxy 50. The routed reference variant also uses proxy 58. Five cycles apply Hadamard (H), controlled-NOT (CX, equivalent to CNOT), and H gates to the data qubit and measure it, with a reset before cycles two through five. Reference preparation selects the ideal all-zero or all-one trace. The target trace is $11111$. Each reference circuit is sampled 20,000 times, and each target circuit 10,000 times. The SPAM scatter plot displays a sampled subset of 10,000 shots per reference class.

QEC-5 is a parity-check workload with $N=5$ rounds. It prepares a Bell pair on data qubits 49 and 50 and measures their $ZZ$ parity through syndrome qubit 51. \textit{Lifting} adds proxy 52. The syndrome qubit and proxy are reset for each round. Even- and odd-parity references train the contextual decoders. Figure~\ref{fig:mcm_workloads}(b) shows the conditional correction: after measuring parity bit $c_j$, an $X$ correction is applied to the first data qubit if $c_j=1$. Here $E_j$ denotes a possible bit flip on that qubit. A single parity check detects odd parity but cannot identify which of two data qubits flipped. The displayed correction therefore assumes the designated error location used in this workload.

The hardware executions defer this correction and retain the uncorrected syndrome history. The target inserts an $X$ error before the fourth round, giving ideal syndrome trace $00011$. To describe the corresponding Pauli frame, let $s_j$ be a decoded syndrome, set $s_0=0$, and define the detection event $e_j=s_j\oplus s_{j-1}$. The accumulated frame bit is
\begin{equation}
  f_N=\bigoplus_{j=1}^{N}e_j=s_N,
  \qquad F_N=X_{49}^{f_N}.
\end{equation}
Thus the inferred correction can be recorded classically instead of applying a conditional gate after each round. In a terminal computational-basis measurement, applying this frame would replace the raw first-data-qubit bit $b_{49}$ by $b_{49}\oplus f_N$. The plotted software-frame criterion checks the inferred frame bit against its ideal value. Because the sequence of exclusive-OR (XOR) operations telescopes, this is the final-syndrome classification metric in Fig.~\ref{fig:combined_workload_results}(b). It does not evaluate a physically corrected Bell state. With active correction, the ideal syndrome would return to zero after the correction, whereas the uncorrected hardware trace remains odd after the injected error.

The teleportation-related circuit uses payload 50, Alice 51, Alice proxy 52, Bob 58, and Bob proxy 71. Figure~\ref{fig:mcm_workloads}(c) shows teleportation of an arbitrary input $\ket{\psi}$. After Bell-pair preparation and the Bell measurement, outcomes $c_{\mathrm{P}}$ and $c_{\mathrm{A}}$ leave Bob in $X^{c_{\mathrm{A}}}Z^{c_{\mathrm{P}}}\ket{\psi}$, up to a global phase. Classical feedforward applies $X$ conditioned on $c_{\mathrm{A}}$ and then $Z$ conditioned on $c_{\mathrm{P}}$, recovering $\ket{\psi}$. Equivalently, the two bits specify a Pauli frame that can be propagated through subsequent Clifford operations and used to reinterpret final measurements. For a computational-basis readout, only the $X$ component changes the measured bit: $b_{\mathrm{out}}=b_{\mathrm{B}}\oplus c_{\mathrm{A}}$. The $Z$ component affects phase and is needed for general quantum-state recovery, but does not change this measurement probability.

In the experiments, reference circuits prepare payload $\ket{0}$ or $\ket{1}$. The target independently acquires the same circuit as the input-$\ket{1}$ reference. Alice--Bob entanglement and the Bell-measurement operations precede Bob's delayed readout. There is no physical outcome-conditioned $X/Z$ correction. Instead, the neural decoder uses the joint Alice--Bob I--Q readout signals, together with the proxy channels specified in Appendix~\ref{app:workload_decoders}, to infer the input bit directly. This learns the computational-basis input classification associated with the frame relation rather than explicitly XORing separately thresholded bits. The payload measurement is acquired but is not a decoder feature, consistent with the phase frame being irrelevant to this basis-specific metric. The experiment therefore evaluates offline input classification, not arbitrary-state teleportation fidelity.

QEC-5 and teleportation acquisitions use a delay of 2000 ns on Fez and 2480 ns on Marrakesh to model the measurement and classical-processing window. Each reference or target circuit has 10,000 shots, collected in five shuffled blocks of 2000. The conditional gates in Fig.~\ref{fig:mcm_workloads}(b,c) specify the corrections accounted for by Pauli-frame tracking in the experiments. \textit{Lifting} auxiliary wires are omitted from the schematics.

\begin{figure*}[t]
  \centering
  \subfigure[QR-Stabilizer.]{\makebox[0.32\textwidth][c]{\includegraphics[height=1.5in]{qr_stab.pdf}}}
  \hfill
  \subfigure[QEC-5.]{\makebox[0.32\textwidth][c]{\includegraphics[height=1.5in]{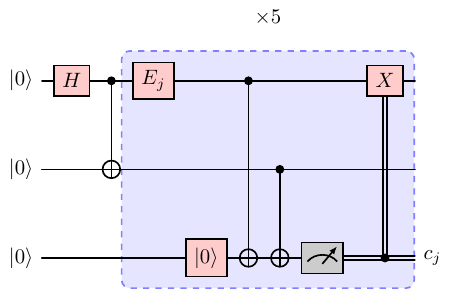}}}
  \hfill
  \subfigure[Teleportation.]{\makebox[0.32\textwidth][c]{\includegraphics[height=1.5in]{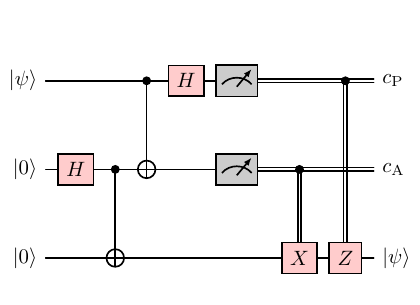}}}
  \caption{MCM workloads. (a) Repeated stabilizer measurement with reset and reuse. (b) Repeated parity extraction with an outcome-conditioned $X$ correction. (c) Teleportation with outcome-conditioned $X$ and $Z$ corrections. Double lines carry classical measurement outcomes. \textit{Lifting} auxiliary wires are omitted. The experiments account for the corrections in (b,c) through the Pauli-frame interpretation described in the text.}
  \label{fig:mcm_workloads}
\end{figure*}

Figure~\ref{fig:workload_topology} shows the physical subgraph used by these layouts. The native connections are $49$--$50$, $50$--$51$, $51$--$52$, $51$--$58$, and $58$--$71$. In particular, qubits 49 and 51 are not directly coupled: their logical parity-check interaction requires routing or gate decomposition through the hardware connectivity. The graph distinguishes native couplings from the logical gates in the circuit schematics.

\begin{figure}[t]
  \centering
  \includegraphics[width=\columnwidth]{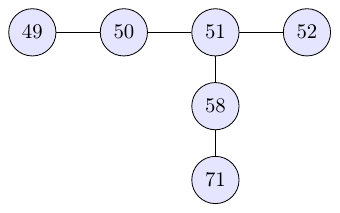}
  \caption{Physical connectivity used by the MCM workloads on \texttt{ibm\_fez}, \texttt{ibm\_kingston}, and \texttt{ibm\_marrakesh}. Nodes identify physical qubits, and edges denote native couplings. The workload-specific qubit assignments are given in the text.}
  \label{fig:workload_topology}
\end{figure}

\subsection{Decoders and Evaluation Populations}\label{app:workload_decoders}

All three MCM workloads---QR-Stabilizer, QEC-5, and teleportation---use the same decoder architecture and training settings. The PyTorch decoder has one hidden layer of 32 neurons with rectified linear unit (ReLU) activations. Single-channel models use one binary head, and \textit{Lifting} uses two heads supervised with labels $00/11$. Reported predictions use the data-qubit head. Training uses Adam with learning rate 0.003, batches of 256, 50 epochs, and an 80/20 shuffled reference split.

For QR-Stabilizer, features are scaled by a power of two derived from their maximum absolute value. Pooled references set the training scale, while each inference batch computes its own scale. Figure~\ref{fig:qr_stab_acc_per_meas} compares target-shot accuracy for NN-based QSD with held-out calibration accuracy for \textit{Lifting} with \textit{in situ} calibration. The two populations are distinct.

For QEC-5 and teleportation, features are standardized using a scaler fitted on pooled references. For QEC-5, NN-based QSD uses the syndrome channel of the encoded target circuit and discards its proxy features, with either static SPAM or \textit{in situ} calibration. \textit{Lifting} uses joint syndrome--proxy references acquired under the corresponding calibration condition. For teleportation, NN-based QSD uses the same four Alice--Bob I--Q features under either calibration condition. Both \textit{Lifting} variants use \textit{in situ} calibration. Adding Alice's proxy gives six inputs, while including both Alice's and Bob's proxies gives eight. The two unfiltered \textit{Lifting} bars therefore compare different input channels rather than an isolated calibration change. In the teleportation panel, the compact legend labels ``Lift'' and ``Lift (in situ)'' identify the Alice-proxy and both-proxy variants, respectively.

In Fig.~\ref{fig:combined_workload_results}(a), NN-based QSD and \textit{Lifting} use single-channel and joint I--Q references, respectively. Each is evaluated with static SPAM calibration and with cycle-matched \textit{in situ} calibration. Circuit success in Fig.~\ref{fig:combined_workload_results} is defined separately for each workload. QR-Stabilizer whole-trace accuracy is the fraction of target shots for which all five predictions equal the ideal trace. QEC-5 accuracy is final-syndrome correctness. The reported software-frame XOR reduces to that final bit. Teleportation-related accuracy is the fraction of input-$\ket{1}$ target shots classified as $1$. QR-Stabilizer filtering fits Gaussian reference distributions with empirical priors and covariance regularizer $10^{-6}I$, and requires posterior confidence at least 0.75 in every cycle. Its coverage is data dependent. The other two workloads rank shots by $|P(Y=1\mid\mathbf{x})-1/2|$ and retain the most confident 90\%, with ties potentially increasing coverage. Filtered accuracies use retained shots as their denominator.

\subsection{Distributional Estimates and Comparator Definitions}

The Bayes-risk curves are equal-prior plug-in estimates from Gaussian kernel densities fitted to normalized I--Q samples. In Fig.~\ref{fig:main_readout}(d), each target-cycle distribution is compared with the SPAM-$0$ reference, so this constructed discrimination task differs from a contemporaneous two-class workload acquisition. The digital-error curve transfers the measured accuracy offset between NN-based QSD and digital readout from a separate Fez acquisition to the plotted target-shot NN-based QSD accuracy. Figure~\ref{fig:bayes_hardware_val} instead compares target-only and joint KDE risk estimates. Its signed histograms come from quadratic discriminant scores.

The distance sweep compares native assignment accuracy, hard decoding of spatial repetition readouts, neural decoding of sequential temporal observations, and neural decoding of joint I--Q readout signals from spatial encoding. Spatial distance increases physical-qubit count, whereas temporal distance increases the number of measurements on one qubit. The fitted accuracy curve uses $\kappa=0.45$ and $\alpha=9.03$ for the displayed Fez configurations. These fitted parameters describe that sweep rather than establish the conditional-overlap assumptions on hardware.

\bibliography{refs}

\end{document}